%% file: arxiv.tex
\documentclass[runningheads,a4paper]{llncs}

\usepackage{amssymb}
\usepackage{amsmath}
\usepackage{graphicx}
\usepackage{algorithmic}
\usepackage{algorithm}
\usepackage{subfigure}
\usepackage{url}
\usepackage{appendix}
\urldef{\mailsa}\path|{fazli, jalaly, s_sadeghian}@ce.sharif.edu| 
\urldef{\mailsb}\path|{jhabibi,ghodsi}@sharif.edu|

\begin{document}

\mainmatter  

\title{On the Non-progressive Spread of Influence through Social Networks}

\titlerunning{} 

%
%
\author{MohammadAmin Fazli\inst{1}
\and Mohammad Ghodsi \inst{1}
\and Jafar Habibi\inst{1} \\
Pooya Jalaly Khalilabadi\inst{1}
\and Vahab Mirrokni\inst{2}
\and Sina Sadeghian Sadeghabad\inst{1}}
\authorrunning{} 
\institute{Computer Engineering Department, Sharif University of Technology, Tehran, Iran\\
\mailsa\\
\mailsb\\
\and Google Research NYC, 76 9th Ave, NewYork, NY 10011\\
}

%


%

\toctitle{} 
\tocauthor{Authors' Instructions}
\maketitle
\thispagestyle{empty}
\setcounter{page}{0}

\begin{abstract}
\input{abstract.tex}
\end{abstract}
\newpage
\input{intro.tex}

\input{nonprogressive.tex}

\input{powerlaw.tex}
\input{experiments.tex}

\input{convergence.tex}
\section{Conclusions}
\input{conclusion.tex}

\section*{Acknowledgments}
The authors are thankful to Soroush Hosseini and Morteza Saghafian for their ideas and their helps.

\input{appendix.tex}

\end{document}

%% file: abstract.tex
The spread of influence in social networks is studied in two main categories: the progressive model and the non-progressive model (see e.g. the seminal work of Kempe, Kleinberg, and Tardos in KDD 2003). While the progressive models are suitable for modeling the spread of influence in monopolistic settings, non-progressive are more appropriate for modeling non-monopolistic settings, e.g., modeling diffusion of two competing technologies over a social network.
Despite the extensive  work on the progressive model, non-progressive models have not been studied well. 
In this paper, we study the spread of influence in the non-progressive model under the strict majority threshold: given a graph $G$ with a set of initially infected nodes, each node 
gets infected at time $\tau$ iff a majority of its neighbors are infected at time $\tau-1$. Our goal in the \textit{MinPTS} problem is to find a minimum-cardinality initial set of infected nodes that would eventually converge to a steady state where all nodes of $G$ are infected.

We prove that while the MinPTS is NP-hard for a restricted family of graphs, it admits an improved constant-factor approximation algorithm for power-law graphs. We do so by 
proving lower and upper bounds in terms of the minimum and maximum degree of nodes in the graph. The upper bound is achieved in turn by applying a natural greedy algorithm.
Our experimental evaluation of the greedy algorithm also shows its superior performance compared to other algorithms for a set of real-world graphs as well as the random power-law graphs. 
Finally, we  study the convergence properties of these algorithms and show that the non-progressive model converges in at most $O(|E(G)|)$ steps. 

%% file: intro.tex
\section{Introduction} \label{s:introduction}
Studying the \textit{spread of social influence} over in networks under various  \textit{propagation models} is a central issue in  social network analysis\cite{freeman2004development,dezső2002halting,pastor2001epidemic,wilson1989levels}. This issue plays an important role in several real-world applications including viral
marketing~\cite{brown1987social,domingos2001mining,richardson2002mining,kempe2003maximizing}. As categorized by Kempe et al.~\cite{kempe2003maximizing}, there are two main types of influence propagation models: the progressive and the non-progressive models.  In \textit{progressive models}, infected (or influenced) vertices will remain infected forever, but in the \textit{non-progressive model}, under some conditions, infected vertices may become uninfected again.
In the context of viral marketing and diffusion of technologies over social networks, the progressive model captures the monopolistic settings where one new service is propagated among nodes of the social network. On the other hand, in non-monopolistic settings, multiple service providers 
might be competing to get people adopting their services, and thus users may switch among two or more services back and forth. As a result, in these non-monopolistic settings, the more relevant model to capture the spread of influence is the non-progressive model~\cite{IKMW07,B93,E93,JY05}. 

 While the progressive model has been studied extensively in the literature \cite{kempe2003maximizing,tang2009social,goyal2010approximation,ben2009exact,chang2009spreading,chen2008approximability,chen2009efficient}, the non-progressive model
has not received much attention in the literature. In this paper, we study the non-progressive influence models, and  report both theoretical and experimental results for our models.  We focus on the {\em the strict majority} propagation rule in which the state of each vertex at time $\tau$ is determined by the state of the majority of its neighbors at time $\tau-1$. 
As an application of this propagation model, consider two competing technologies (e.g. IM service) that are competing in attracting nodes of a social network to adopt their service, and 
nodes tend to adopt a service that the majority of their neighbors already adopted. This type of influence propagation process can be captured by applying the strict majority rule.  
Moreover, as an illustrative example of the linear threshold model~\cite{kempe2003maximizing},  the strict majority propagation model is suitable for modeling transient faults in fault tolerant systems~\cite{flocchini2003time,peleg2002local,flocchini2004dynamic}, and also used in verifying convergence of consensus problems on social networks \cite{mossel2009reaching}.  Here we study the non-progressive influence models under the strict majority rule. In particular, we are mainly interested in {\em minimum perfect target set} problem where the goal is to identify a target set of nodes to infect at the beginning of the process so that all nodes get infected at the end of the process. We will present  approximation algorithms and hardness results for the problem as well experimental evaluation of our results.  As our main contributions, we provide  improved upper and lower bounds on the size of the minimum perfect target set, which in turn, result in improved constant-factor approximations for power-law graphs. Finally, we also study the convergence rate of our algorithms and report preliminary results. Before stating our results, we define the problems and models formally.


\noindent {\bf Problem Formulations.} Consider a graph $G(V, E)$. Let $N(v)$ denote the set of neighbors of vertex $v$, and $d(v)=|N(v)|$. Also, let $\Delta(G)$ and $\delta(G)$ denote the maximum and minimum degree of vertices in $G$  respectively. 
The induced subgraph of $G$ with a vertex set $S \subseteq V(G)$ is denoted by $G[S]$. Also $d_{S}(v)$ denotes the number of neighbors of $v$ in subset $S$.  

A \textit{0/1 initial assignment} is a function $f_0:V(G)\rightarrow \{0,1\}$. For any 0/1 initial assignment $f_0$, let $f_{\tau}:V(G)\rightarrow \{0,1\}$ ($\tau \geq 1$) be the state of vertices at time $\tau$ and $t(v)$ be the threshold associated with vertex $v$. For the strict majority model, the threshold $t(v)= \lceil \frac{d(v)+1}{2}\rceil$ for each vertex $v$.

In the \textit{non-progressive strict majority} model:
$$
f_{\tau}(v) = \left\{ \begin{array}{ll}
	0 & \textrm{if $\sum_{u\in N(v)} f_{\tau-1}(u) < t(v)$}\\
	1 & \textrm{if $\sum_{u\in N(v)} f_{\tau-1}(u) \geq t(v)$ . }\\
\end{array} \right.
$$
In \textit{progressive strict majority} model:
$$
f_{\tau}(v) = \left\{ \begin{array}{ll}
	0 & \textrm{if $f_{\tau - 1}(v)=0$ and $\sum_{u\in N(v)} f_{\tau-1}(u) < t(v)$}\\
	1 & \textrm{if $f_{\tau - 1}(v)=1$ or $\sum_{u\in N(v)} f_{\tau-1}(u) \geq t(v)$ . }\\
\end{array} \right.
$$
Strict majority model is related to the \textit{linear threshold model} in which $t(v)$ is chosen at random  and not necessarily equal to $\lceil \frac{d(v)+1}{2}\rceil$. 

A 0/1 initial assignment $f_0$ is called a \textit{perfect target set} (PTS) if for a finite $\tau$, $f_{\tau}(v)=1$ for all $v\in V(G)$, i.e., the influence will converge to 
a steady state of all $1$'s.  The cost of a target set $f_0$, denoted by $cost(f_0)$, is the number of vertices $v$ with $f_0(v)=1$. The minimum \textit{perfect target set }(MinPTS) problem is to find a perfect target set with minimum cost. The cost of this minimum PTS is denoted by $PPTS(G)$ and $NPPTS(G)$ respectively for progressive and non-progressive models. This problem is also called \textit{target set selection} \cite{ackerman2010combinatorial}. Another variant of this problem is the 
 \textit{maximum active set} problem~\cite{ackerman2010combinatorial} where 
 the goal is to find at most $k$ nodes to activate (or infect) at time zero such that the number of
 finally infected vertices is maximized. 

A graph is power-law if and only if its degree distribution follows a power-law distribution asymptotically. That is, the fraction $P(x)$ of nodes in the network having degree $x$ goes for large number of nodes as $P(x) = \alpha x^{-\gamma}$ where $\alpha$ is a constant and $\gamma > 1 $ is called power-law coefficient. It is widely observed that most social networks are power-law~\cite{clauset2009power}.   
 
\noindent {\bf Our Results and Techniques.}
In this paper, we study the spread of influence in the non-progressive model under the strict majority threshold.  
We present  approximation algorithms and hardness results for the problem as well experimental evaluation of our results.  As our main contributions, we provide  improved upper and lower bounds on the size of the minimum perfect target set, which in turn, result in improved constant-factor approximations for power-law graphs. In addition, we prove that the MinPTS problem (or computing $NPPTS(G)$) is NP-hard for a restricted family of graphs. In particular, we prove lower and upper bounds on $NPPTS(G)$ in terms of the minimum degree ($\delta(G)$) and maximum degree ($\Delta(G)$) of nodes in the graph, i.e., we show that $\frac{2n}{\Delta(G)+1} \leq NPPTS(G) \leq \frac{n\Delta(G)(\delta(G)+2)}{4\Delta(G)+(\Delta(G)+1)(\delta(G)-2)}$.  The proofs of these bounds are combinatorial and start by observing that
in order to bound $NPPTS(G)$ for general graphs, one can bound it for bipartite graphs. The upper bound is achieved in turn by applying a natural greedy algorithm which can be easily implemented. 
Our experimental evaluation of the greedy algorithm also shows its superior performance compared to other algorithms for a set of real-world graphs as well as the random power-law graphs. 
Finally, we study the convergence properties of these algorithms. We  first observe that the process will always converges  to a fixed point or a cycle of size two. Then we focus on the convergence time and prove that for a given graph $G$, it takes at most $O(|E(G)|)$ rounds for the process to converge.  We also evaluate the convergence rate of the non-progressive influence models on some real-world social networks, and report the average convergence time for a randomly chosen set of initially infected nodes.

{\bf \noindent More Related Work.}
The non-progressive spread of influence under the strict majority rule is related to the diffusion of two or more competing technologies over 
a social network~\cite{IKMW07,B93,E93,JY05}. 
As an example, an active line of research in economics and mathematical sociology is concerned with modeling these types of diffusion processes as a coordination game played on a social network~\cite{IKMW07,B93,E93,JY05}. Note that none of these previous prior work provide a bound for the perfect target set problem.

It has been brought to our attention that in a relevant unpublished work by Chang \cite{chang2010reversible}, the MinPTS problem on pawer-law graphs is studied and the bound of $NPPTS(G) = O(\lceil \frac{|V|}{2^{\gamma-1}}\rceil)$ is proved  under non-progressive majority models in a power-law graph. But his results do not practically provide any bound for the strict majority model. We will show that our upper bound is better and practically applicable for different amounts of $\gamma$ under strict majority threshold.
 
Tight or nearly tight bounds on the $PPTS(G)$ are known for special types of graphs such as torus, hypercube, butterfly and chordal rings \cite{flocchini2001optimal,flocchini2003time,luccio1999irreversible,peleg2002local,pike2005decycling}. The best bounds for progressive strict majority model in general graphs are due to Chang and Lyuu. In  \cite{chang2009spreading}, they showed that for a directed graph $G$, $PPTS(G) \leq  \frac{23}{27}|V(G)|$. In \cite{chang2009irreversible}, they improved their upper bound to $\frac{2}{3}|V(G)|$ for directed graphs and $\frac{|V(G)|}{2}$ for undirected graphs.  However, to the best of our knowledge, there is no known bound for $NPPTS(G)$ for any type of graphs. In this paper, we will combinatorially prove that $\frac{2n}{\Delta(G)+1} \leq NPPTS(G) \leq \frac{n\Delta(G)(\delta(G)+2)}{4\Delta(G)+(\Delta(G)+1)(\delta(G)-2)}$.

It is known  that the Target Set Selection problem and Maximum Active Set problem are both NP-hard in the linear threshold model~\cite{kempe2003maximizing}, and approximation algorithms have been developed for these problems. 
 Kempe et al. \cite{kempe2003maximizing} and Mossel and Roch \cite{mossel2007submodularity} present a $(1 - \frac{1}{e})$-approximation algorithm for the maximum active set problem by showing that the set of finally influenced vertices as a function of the originally influenced nodes is submodular. On the other hand, it has been shown that the target set selection problem is not approximable for different propagation models~\cite{goyal2010approximation,ben2009exact,chang2009irreversible,chen2008approximability}. The inapproximability result of Chang and Lyuu in \cite{chang2009irreversible} on progressive strict majority threshold model is the most relevant result to our results. They show that  unless $NP \subseteq TIME(n^{O(\ln{\ln{n}})})$, no polynomial time $((1/2 - \epsilon) \ln{|V|})$-approximation algorithm exists for computing $PPTS(G)$. To the best of our knowledge, no complexity theoretic results have been obtained for the non-progressive models.

The problem of maximizing social influence for specific family of graphs has been  studied under the name of~\textit{dynamic monopolies} in the combinatorics literature~\cite{flocchini2001optimal,flocchini2003time,luccio1999irreversible,peleg2002local,pike2005decycling,chang2009irreversible,ackerman2010combinatorial,chang2010reversible}. All these results are for the progressive model.
The optimization problems related to the non-progressive influence models are not well-studied in the literature. The one result in the area is due to   Kempe et al. \cite{kempe2003maximizing} who presented a general reduction from non-progressive models to progressive models. Their reduction, however, is not applicable to the perfect target set selection problem.


%% file: nonprogressive.tex
\section{Non-Progressive Spread of Influence in General Graphs} \label{sec:general}
In this section,  we prove  lower bound and upper bounds for minimum PTS in graphs, and finally show that finding the minimum PTS in general graphs is NP-Hard.

\noindent {\bf Lower bound.} 
\input{glowerbound.tex}

\noindent {\bf Upper bound.} 
\input{gupperbound.tex}
\noindent {\bf NP-Hardness.} 
\input{gnphard.tex}

%% file: glowerbound.tex
The following theorem shows that if we have some lower bound and upper bound for minimum Perfect Target Set in bipartite graphs then these bounds could be generalized to all graphs ( Theorem \ref{thm:lowernppts}).
\begin{lemma}\label{thm:bip_to_gen}
If $\alpha |V(H)| \leq NPPTS(H)\leq \beta |V(H)| $ for every bipartite graph $H$ under strict majority threshold, then $\alpha |V(G)| \leq NPPTS(G)\leq \beta |V(G)|$ under strict majority threshold for every graph $G$ (see appendix \ref{app:general_proofs}).
\end{lemma}

The following lemma shows characteristics of PTSs in some special cases. These will be used in proof of our theorems.
\begin{lemma}\label{lem:maxdegree}
Consider the non-progressive model and let $G=(X,Y)$ be a bipartite graph and $f_0$ be a perfect target set under strict majority threshold. For every $S\subseteq V(G)$ if $\sum_{v \in S \cap X}{f_0(v)} = 0$ or $\sum_{v\in S \cap Y}{f_{0}(v)} = 0$, then there exists at least one vertex $u$ in $S$ such that $d_S(u) \leq d(u) - t(u)$ (see appendix \ref{app:general_proofs}). 
\end{lemma}
If the conditions of previous lemma holds, we can obtain an upper bound for number of edges of the graph. Following lemma provides this upper bound. This will help us finding a lower bound for NPPTS of graphs. The function $t:V(G)\rightarrow \mathbb{N}$ may be any arbitrary function but here it is interpreted as the threshold function.
\begin{lemma}\label{lem:maxedge}
Consider a graph $G$ with $n$ vertices. If for every $S\subseteq V(G)$ there exists at least one vertex $v$ for which $d_S(v) \leq d(v) - t(v)$, then $|E(G)| \leq \sum_{u\in V(G)} (d(u) - t(u))$ (see appendix \ref{app:general_proofs}).
\end{lemma}
The following theorem shows that for every bipartite graph $G$, $NPPTS(G) \geq \frac{2|V(G)|}{\Delta(G)+1}$. Lemma \ref{thm:bip_to_gen} generalizes this theorem to all graphs. Also, Theorem \ref{thm:tightlower} shows that this bound is tight. In the following, the induced subgraph of $G$ with a vertex set $S \subseteq V(G)$ is denoted by $G[S]$.
\begin{theorem}\label{thm:lowernppts}
For every bipartite graph $G=(X,Y)$ of order $n$, $NPPTS(G) \geq \frac{2n}{\Delta(G)+1}$.
\end{theorem}
\begin{proof}
Let $f_0$ be an arbitrary $PTS$ for $G$. Partition the vertices of graph $G$ into three subsets $B_X$, $B_Y$ and $W$ as follow.
\begin{equation*}
\begin{split}
B_X & =\{v \in X \, | \, f_0(v)=1 \} \\
B_Y & =\{v \in Y \, | \, f_0(v)=1 \} \\
W & =\{v \in V(G) \, | \, f_0(v)=0 \} \\
\end{split}
\end{equation*}
Consider the induced subgraph of $G$ with vertex set $B_X \cup W$ and suppose that $S\subseteq B_X \cup W$. For every vertex $v\in Y \cap S$, we have $f_0(v)=0$. So By Lemma \ref{lem:maxdegree}, for every $S\subseteq B_X \cup W$ there is at least one vertex $u$ such that $d_S(u)\leq d(u)-t(u)$.  By Lemma \ref{lem:maxedge}, this implies that $G[B_X \cup W]$ has at most $\sum_{u\in B_X\cup W} (d(u)-t(u))$ edges. Similarly we can prove that $G[B_Y \cup W]$ has at most $\sum_{u\in B_Y\cup W} (d(u) - t(u))$ edges. Let $e_{W}$ be the number of edges in $G[W]$, $e_{WX}$ be the number of edges with one end point in $B_X$ and the other end point in $W$ and $e_{WY}$ be the number of edges with one end point in $B_Y$ and the other end point in $W$. we have:
\begin{equation*}
\begin{split}
e_{WX} + e_{W} \leq \sum_{v\in B_X \cup W} (d(v) - t(v)) \\
e_{WY} + e_{W} \leq \sum_{v\in B_Y \cup W} (d(v) - t(v)) \\
\end{split}
\end{equation*}
and so,
\begin{equation*}
\begin{split}
e_{WX}+e_{WY}+2e_{W} \leq \sum_{v\in V(G)} (d(v) - t(v)) + \sum_{v\in W} (d(v) - t(v)) 
\end{split}
\end{equation*}
The total degree of vertices in $W$ is $\sum_{v\in W}d(v)=e_{WX}+e_{WY}+2e_{W}$. Thus
\begin{equation*}
\begin{split}
\sum_{v\in W}d(v) \leq \sum_{v\in V(G)} (d(v) - t(v)) + \sum_{v\in W} (d(v) - t(v))
\end{split}
\end{equation*}
If we denote the set of vertices for which $f_0$ is equal to $1$ by $B$, we have
\begin{equation}\label{eq:lower}
\sum_{v\in W}(2t(v) - d(v)) \leq \sum_{v\in B}(d(v)-t(v))
\end{equation} 
For every vertex $v$, $t(v)\geq\frac{d(v)+1}{2}$, so
\begin{equation*}
\begin{split}
& |W| \leq \sum_{v\in B}\frac{d(v)-1}{2} \Rightarrow |W| \leq \frac{\Delta -1}{2}(|B|) \\
& \Rightarrow |B| \geq \frac{2n}{\Delta +1} \\
\end{split}
\end{equation*}
And the proof is complete.
\end{proof}

We now show that the bound in Theorem \ref{thm:tightlower} is tight.
\begin{lemma} \label{thm:tightlower}
For infinitely many $n$'s, there exists a $2d+1$-regular graph with $n$ vertices such that $NPPTS(G)= \frac{n}{d+1}$ under strict majority rule (see appendix \ref{app:general_proofs}).
\end{lemma}

%% file: gupperbound.tex
In this section, we present a greedy algorithm 
which gives an upper bound for $NPPTS(G)$. 
\begin{algorithm}
\caption{Greedy NPPTS}\label{alg:greedy}
\begin{algorithmic}
\STATE sort the vertices in $G$ in ascending order of their degrees as the sequence $v_1,\ldots , v_n$.
\FOR{$i=1$ to $n$}
	\STATE $\mathtt{whiteadj}[v_i]=0$
	\STATE $\mathtt{blocked}[v_i]=0$
\ENDFOR
\FOR{$i=1$ to $n$}
\FOR {each $u \in N(v_i)$}
	\IF {$\mathtt{whiteadj}[u]=d(u)-t(u)$}
		\STATE $\mathtt{blocked}[v_i]=1$
	\ENDIF
\ENDFOR
\IF {$(\mathtt{blocked}[v_i] = 1)$}
        \STATE $f_0(v)=1$
\ELSE
        \STATE $f_0(v)=0$
       	\FOR {each $u \in N(v_i)$}
       		\STATE $\mathtt{whiteadj}[u]+=1$
       	\ENDFOR
\ENDIF
\ENDFOR

\end{algorithmic}
\end{algorithm}
\begin{theorem}\label{thm:uppernppts}
For every graph $G$ of order $n$, $NPPTS(G)\leq \frac{n\Delta(\delta+2)}{4\Delta+(\Delta+1)(\delta-2)}$ under strict majority threshold.
\end{theorem}
Algorithm \ref{alg:greedy} guarantees this upper bound. This algorithm gets a graph $G$ of order $n$ and the threshold function $t$ as input and determines the values of $f_0$ for each vertex. 
\begin{lemma} \label{lem:greedy_correctness}
The algorithm Greedy NPPTS finds a Perfect Target Set for non-progressive spread of influence. 
\end{lemma}
\begin{proof}
By induction on the number of vertices for which $f_0$ is determined, we prove that $f_0$ remains a PTS after each step of algorithm if we assume that $f_0$ is $1$ for undetermined values.
It is clear that the claim is true at the beginning. Consider a set of values of $f_0$ which forms a PTS and let $v$ be a vertex for which value of $f_0(v)$ is set to $0$ by the algorithm in the next step. By induction hypothesis, $f_0$ is a PTS if $f_0(v)$ is assumed to be $1$. According to the algorithm, $f_0(v)$ is set to $0$ iff the value of $\mathtt{blocked}[v]$ is zero i.e. no adjacent vertex of $v$, say $u$, has exactly $d(u)-t(u)$ adjacent initially uninfected vertices. So by setting $f_0(v)$ to $0$, each initially infected vertex $w$ still has at least $t(w)$ infected vertices and also $v$ has at least $t(v)$ initially infected neighbors itself. Thus, after one step of propagation, all initially infected vertices plus $v$ are infected and by induction hypothesis, all vertices will be infected eventually and so $f_0$ remains a PTS.

\end{proof}
\begin{lemma} \label{lem:gupperbound}
 For every graph $G$ of order $n$, Greedy NPPTS guarantees the upper bound of $\frac{n\Delta(\delta+2)}{4\Delta+(\Delta+1)(\delta-2)}$ for $NPPTS(G)$ under strict majority threshold where $\Delta$ and $\delta$ are maximum and minimum degree of vertices respectively. 
\end{lemma}
\begin{proof}
According to the algorithm, for each vertex $v$, the value of $f_0(v)$ is set to $1$ iff $\mathtt{whiteadj}[u]=d(u)-t(u)$ for some $u\in N(v)$. Let $S$ be the set of vertices $u$ for which $\mathtt{whiteadj}[u]=d(u)-t(u)$. $B$ and $W$ denote the set of infected and uninfected vertices respectively. We have:

\begin{equation*}
\begin{split}
& \sum_{v\in S}(d(v) - t(v))\leq \sum_{v\in W}d(v) \Rightarrow \sum_{v\in S}(\frac{d(v)}{2}-1) \leq \sum_{v\in W}d(v)  \\
\end{split}
\end{equation*}
Therefore, 
\begin{equation*}
\begin{split}
&  (\frac{\delta}{2}-1)|S| \leq \Delta|W|. \Rightarrow |S|\leq \frac{2\Delta}{\delta-2}|W| \\\
\end{split}
\end{equation*}
Each vertex in $B$ has at least one adjacent vertex in $S$ and each vertex $v\in S$ has at least $d(v)-t(v)$ adjacent edges to $W$ and so at most $t(v)$ adjavent edges to $B$, thus:
\begin{equation*}
\begin{split}
 |B| & \leq \sum_{v\in S}(t(v)) \leq \sum_{v \in S}(\frac{d(v)}{2}+1) \leq 2|S|+\sum_{v\in W}d(v) \\
& \leq 2|S| + \Delta|W| \leq (2\frac{2\Delta}{\delta-2}+\Delta)|W| \leq \frac{4\Delta+\Delta(\delta-2)}{\delta-2}|W| \\
\end{split}
\end{equation*}
Thus,
\begin{equation*}
\begin{split}
&  |B|\leq \frac{\Delta(\delta+2)}{4\Delta+(\Delta+1)(\delta-2)}n \\
\end{split}
\end{equation*}

\end{proof}

The approximation factor of the algorithm follows from previous claim and the lower bound provided by Theorem \ref{thm:lowernppts}:

\begin{corollary}
The Greedy NPPTS algorithm is a $\frac{\Delta(\Delta+1)(\delta+2)}{8\Delta+2(\Delta+1)(\delta-2)}$ approximation algorithm for $NPPTS$ problem.
\end{corollary}

%% file: gnphard.tex
In this section, we use a reduction from the Minimum Dominating Set problem (MDS) \cite{allan1978domination} to prove the NP-hardness of computing $NPPTS(G)$.
The proof of following theorem is provided in appendix \ref{app:general_proofs}.
\begin{theorem} \label{thm:nphardness}
If there exists a polynomial-time algorithm for computing $NPPTS(G)$ for a given graph $G$ under the strict majority threshold, then $P = NP$. 
\end{theorem}

%% file: powerlaw.tex
\section{Non-Progressive Spread of Influence in Power-law graphs}

In this section, we investigate the non-progressive spread of influence in power-law graphs, and show that the greedy algorithm presented in the previous section is indeed a constant-factor approximation algorithm for power-law graphs. 
For each natural number $x$, we assume that the number of vertices with degree $x$ is proportional to $x^{-\gamma}$ and use $\alpha$ as the normalization coefficient. The value of $\gamma$, known as power-law coefficient, is known to be between $2$ and $3$ in real-world social networks 
. We denote the number of vertices of degree $x$ by $P(x)=\alpha x^{-\gamma}$. Let $n$ be the number of vertices of graph, so we have:
\begin{equation*}
\begin{split}
&  n = \sum_{x= 1}^{\infty} \alpha x^{-\gamma} = \alpha \zeta(\gamma) \Rightarrow \alpha = \frac{n}{\zeta(\gamma)}, \\
\end{split}
\end{equation*}

\noindent where $\zeta$ is the Riemann Zeta function \cite{ivic1985riemann}. 


\noindent {\bf Lower bound.} 
Consider a power-law graph $G$ with a threshold function $t$ and a perfect target set $f_0$. Denoting the set of initially influenced vertices by $B$ and the rest of the vertices by $W$, from the Equation \ref{eq:lower}, we have:
\begin{equation*}
\begin{split}
& \sum_{v\in W}(2t(v) - d(v)) \leq \sum_{v\in B}(d(v)-t(v)) . \\
\end{split}
\end{equation*}
The maximum cardinality of $W$ is achieved when the degree of all vertices in $B$ is greater than or equal to the degree of all vertices in $W$. In this case, assume that the minimum degree of vertices in $B$ is $k$ and $0\leq p \leq 1$ is the proportion of the vertices of degree $k$ in $B$, so under strict majority threshold we have:
\begin{equation*}
\begin{split}
& \sum_{x=1}^{k-1} \alpha x^{-\gamma} +(1-p)\alpha k^{-\gamma} \leq |W| \leq \sum_{v\in W}(2t(v) - d(v))  \\
& \leq \sum_{v\in B}(d(v)-t(v)) \leq \sum_{x=k+1}^{\infty} \alpha x^{-\gamma}(\frac{x-1}{2}) + p\alpha k^{-\gamma}\frac{k-1}{2} \\
\Rightarrow &  \sum_{x=1}^{k-1} x^{-\gamma} +(1-p)k^{-\gamma} \leq \frac{\sum_{x=k+1}^{\infty} (x^{1-\gamma} - x^{-\gamma}) +pk^{-\gamma}(k-1)}{2} \\
\Rightarrow &  \zeta(\gamma)-\zeta(\gamma,k-1) + (1-p)k^{-\gamma} \\
& \leq \frac{\zeta(\gamma-1,k) - \zeta(\gamma,k)+pk^{-\gamma}(k-1)}{2}. \\
\end{split}
\end{equation*}
By estimating the value of Riemann Zeta function, we can estimate the upper bound of $k$ and lower bound of $p$ for that $k$ to provide a lower bound for $|B|$. Assuming that we have the maximum possible value of $k$ and minimum value of $p$ for that $k$, then:
\begin{equation*}
\begin{split}
& |B| \geq \sum_{x=k+1}^{\infty}\alpha x^{-\gamma} + \alpha p k^{-\gamma} = \frac{\zeta(\gamma,k)+pk^{-\gamma}}{\zeta(\gamma)}n. \
\end{split}
\end{equation*}
The estimated values of lower bound  for $2 \leq \gamma \leq 2.8$ is shown in Figure \ref{fig:chart}.

\noindent {\bf Upper bound} 
Suppose that one has run Greedy NPPTS algorithm under strict majority threshold on a graph with power-law degree distribution. The following theorem shows that unlike general graphs, the Greedy NPPTS algorithm guarantees a constant factor upper bound on power-law graphs.

\begin{theorem} Algorithm Greedy NPPTS initially influences at most $(1+\frac{1}{2^{\gamma+1}}-\frac{1}{2\zeta(\gamma)})n$ vertices under the 
strict majority threshold on a power-law graphs of order $n$.
\end{theorem}
\begin{proof}
We may assume that the input graph is connected. We prove that the number of uninfected vertices of degree $1$ are sufficient for this upper bound. Let $v$ be a vertex of degree more than $1$ with $k$ adjacent vertices of degree $1$ say $u_1,u_2\ldots u_k$.  If $d(v)$ is odd, it is clear that at least $\frac{k}{2}$ of the vertices $u_1,u_2\ldots u_k$ will be uninfected since $k\leq d(v)$. Note that according to the greedy algorithm, the value of $f_0$ for degree $1$ vertices are determined before any other vertex. If $d(v)$ is even, at least $\frac{k}{2}-1$ of vertices $u_1,u_2\ldots u_k$ will be uninfected.  Therefore we have:
\begin{equation*}
\begin{split}
 NPPTS(G) & \leq n-\frac{1}{2}(P(1)-\sum_{x=1}^{\infty}P(2x)) \\
& \leq n-\frac{1}{2}(\alpha\frac{1}{1^\gamma}-\alpha\sum_{x=1}^{\infty}\frac{1}{(2x)^{\gamma}}) \\
& = n-\frac{\alpha}{2}(1-\frac{1}{2^\gamma}\zeta(\gamma)) = n(1+\frac{1}{2^{\gamma+1}}-\frac{1}{2\zeta(\gamma)}) \\
\end{split}
\end{equation*}
\end{proof}

By previous lemma, we conclude that the Greedy NPPTS algorithm is a constant-factor approximation algorithm on power-law graphs under strict majority threshold. The lower bound and upper bound for different values of $\gamma$ are shown in Figure \ref{fig:chart}. As you can see our algorithm acts optimally on social networks with large value of power-law coefficient since upper and lower bound diagram meet each other for these values of power-law coefficient.

\begin{figure}
\begin{center}
\includegraphics[width=5cm, trim =7cm 16.5cm 6cm 4.5cm]{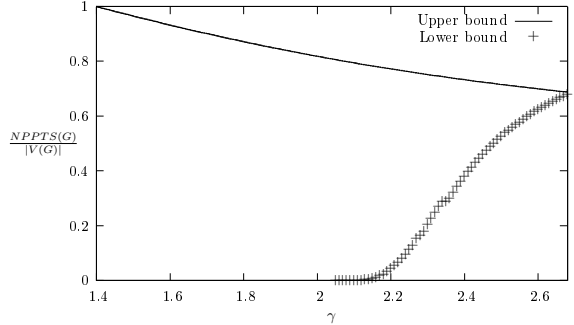}
\end{center}
\caption{Values of upper bound and lower bound in power-law graphs}
\label{fig:chart}
\end{figure} 

%% file: experiments.tex
\section{Experimental Evaluations} \label{s:experim}
In this section we run our algorithm on real-world social networks as well as random power-law graphs with a wide range of power-law coefficients. Following the method used in 
~\cite{kempe2003maximizing}, we compare the performance of our algorithm to other 
heuristics for identifying influential individuals.  

\noindent {\bf Random power-law and real-world networks.} 
We evaluate the performance of the greedy  algorithm on graphs with various amount of power-law coefficient. Following a previously developed way of generating power-law graphs from \cite{aiello2000random}, we set two parameters $\alpha$ and $\gamma$ defined as follows: $\alpha$ is the logarithm of the graph size and $\gamma$ is the log-log growth rate (power-law coefficient). 
For details of the way to generate random power-law graphs, see Appendix ~\cite{experiments}.
We also  run our algorithms over four social networks' data: Who-trusts-whom network of Epinions.com, Slashdot social network, collaboration network of Arxiv Astro Physics, Arxiv High Energy Physics paper citation network, Amazon product co-purchasing network. 
In cases where graph is not connected we select graphs' giant component.


\noindent {\bf Setup.}  We compare our greedy
algorithm with heuristics based on nodes' degrees and centrality
within the network, as well as the  baseline of choosing random nodes to target.
High-Degree and distance-centrality heuristics choose vertices in the order of decreasing degree and decreasing average distance to other nodes. These heuristics are
commonly used in the social science literature as estimates of a node's influence in the social network~\cite{wasserman1994social,kempe2003maximizing}.


In each of these cases, in each step, we check whether the selected vertices are a perfect target set or not. This can be easily verified by simulating spread of influence process until the states of vertices become stable. The simulation process ends at a polynomially bounded time $\tau$ when for each $v \in V(G)$ we have $f_{\tau}(v) = f_{\tau-2}(v)$ (see Theorem \ref{t:convergence} and Theorem \ref{thm:convergence}). 

Notice that because the optimization problem
is NP-hard (Theorem \ref{thm:nphardness}), and the testbed graphs are prohibitively large, we
are not able to compute the optimum value to verify the actual quality of
approximations.

\noindent {\bf Experimental Results.} Figure \ref{f:randompowerlaw} shows the performance of our algorithm in comparison to introduced heuristics on random power-law graphs. For any value of $\gamma$ (power-law coefficient), all heuristics pick almost entire vertices of the graph while our algorithm pick a number of them between proved lower-bound and upper-bound. 
The same phenomena happens for the four real-world social networks data. The results are depicted in Figure \ref{f:realnetworkdatachart}. 

\begin{figure}[htbp]
\centering
\subfigure[Results on the random power-law graphs]{
\includegraphics[width=6cm]{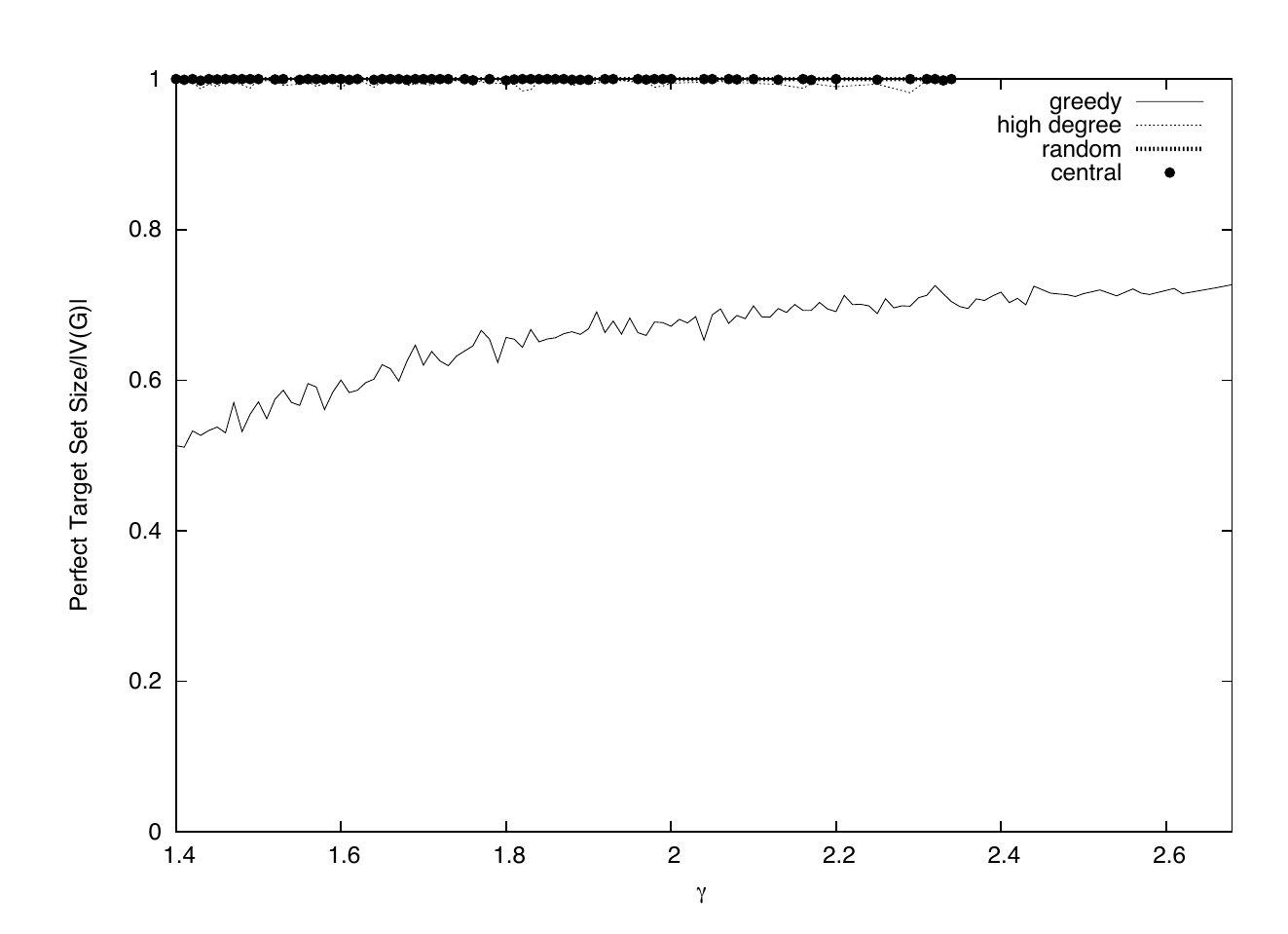}
\label{f:randompowerlaw}
}
\subfigure[Results on the real network data]{
\includegraphics[width=6cm]{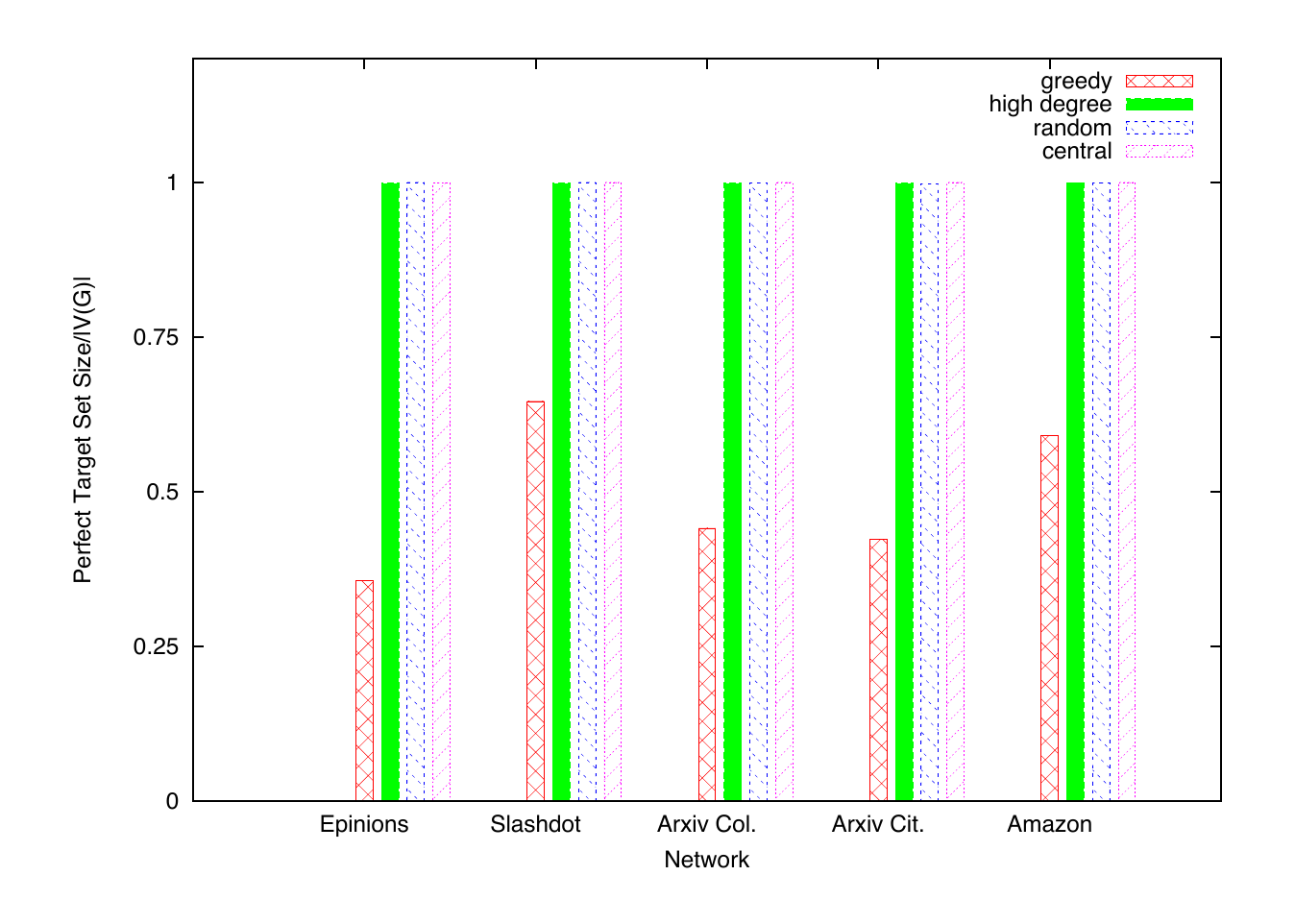}
\label{f:realnetworkdatachart}
}

\end{figure} 



%% file: convergence.tex
\section{Convergence Issues} \label{sec:convergence}
Let the state graph $H$ of a non-progressive spread of influence process for graph $G$ be as follows: Each node of this graph represents one of  possible states of the graph. An edge between two states $A$ and $B$ in $H$ models the fact that applying one step of the influence process on 
state $A$  changes the state to state $B$.
First of all, one can easily see that the non-progressive model may not result to a singleton steady state for any dynamics.  To see this, consider the following example:  a cycle with $2k$ vertices $C = v_{1}v_{2}...v_{2k}$ and at time 0 infect vertices with odd indices. In this case, the process will oscillate  between exactly two states. In fact, one can show a general theorem that any dynamics will converge to either one or two states:

\begin{theorem} \label{t:convergence}
The non-progressive spread of influence process on a graph reaches a cycle of length of at most two (see appendix \ref{app:convergence}).
\end{theorem} 

Using this intuition, one can define the  convergence time of a non-progressive influence process under the strict majority rule  as the time it takes to converge to a cycle of size of two states, i.e., the convergence time is the minimum time $T$ at which  $f_{T}(v) = f_{T+2}(v)$ for all vertices $v \in V(G)$. For a set $S$  of initially infected vertices, let $ct_{G}(S)$ to be the convergence time of the non-progressive process under the strict majority model($T$). In the following theorem, we formally prove an upper bound of $O(|E(G)|)$ for this convergence time:

\begin{theorem} \label{thm:convergence}
For a given graph $G$ and any set $S \subseteq V(G)$, we have $ct_{G}(S) = O(|E(G)|)$. 
\end{theorem}
\begin{proof}
In each time step $\tau$ of the  non-progressive spread of influence, all the vertices apply the function $f_{\tau}$ {\em concurrently}. In order to prove the theorem for such {\em concurrent dynamics} , we first define a simplified {\em sequential dynamics},  prove the convergence time for this simplified dynamics, and finally  give a reduction from the concurrent to the sequential dynamics. In sequential dynamics, the vertices apply the influence process one by one in a sequence of rounds, where in each step one vertex applies the influence process exactly once. 

We first show that the sequential dynamic on every graph $G$ and under the strict majority model converges after at most $O(|E(G)|.|V(G)|)$ steps. To see this bound, consider the following potential function for a graph $G$:  the number of edges whose endpoints have different states. One can see that whenever a vertex changes its state from uninfected to infected the potential of $G$ will  decrease at least one and otherwise it remains unchanged. Consider a vertex which has $k$ state changes during the process until it final convergence. At least $k/2$ of these changes were from uninfected state to infected and so they cause one decrement in the potential function. The initial amount of $G$'s potential is at most $|E(G)|$ and in each step (or $|V(G)|$ consecutive steps),  we have at least one state change. So after at most $2|E(G)|.|V(G)|$ steps the potential of $G$ would reach its minimum, and the proof for the sequential dynamics is complete.

Now using the above observation, we show that the concurrent  dynamics convergences fast. Consider  graph $H = (X,Y)$ built from $G$ in Lemma \ref{thm:bip_to_gen}. We show that for every concurrent dynamics in $G$ with convergence time of $T$, there is an equivalent sequential dynamics in $H$ with convergence time of $c|V(G)|T$ for some constant $c$. This will prove $ct_{G} \in O(|E(G)|)$, since 
we know that the convergence time of the sequential dynamic in graph $H$ is at most $2|V(H)|.|E(H)| = 8|V(G)|.|E(G)| = c|V(G)|.T$. So $T \in O(|E(G)|)$. The main claim follows from  the proof of Lemma \ref{thm:bip_to_gen}. By induction on the number of steps, we can show that the state of vertices in $G$ is equal to the state of vertices in $X$ at odd steps and is equal to the state of vertices in $Y$ at even steps (as we did in the proof of Lemma \ref{thm:bip_to_gen}). Now order vertices of $X$ and $Y$ with numbers $1, 2, \cdots, |V(G)|$ and from $|V(G)|+1$ to $2|V(G)|$. It is easy to see that the sequential dynamics with this ordering, after $|V (H)|$ steps, has the same outcome under the concurrent dynamics in this graph.
\end{proof}

The above theorem is tight i.e. there exists a set of graphs and initial states with convergence time of $O(|E(G)|)$. In power-law graphs since average degree is constant, the number of edges is $O(|V|)$ and thus the convergence time of these graphs is $O(|V|)$. 

Finally, we study convergence time of non-progressive dynamics  on several real-world graphs, and observe the fast convergence of such dynamics on those graphs. See Appendix \ref{app:convergence} for details.

%% file: conclusion.tex
In this paper, we study the minimum target set selection problem in the non-progressive influence model under the strict majority rule and provide theoretical and practical results
for this model.  Our main results include upper bound and  lower bounds for these graphs, hardness and approximation algorithm for this problem. We also apply our techniques on power-law graphs and derive improved constant-factor approximation algorithms for this kind of graphs.

An important follow-up work is to study the minimum perfect set problem for  non-progressive models under other influence propagation rules, e.g. the general linear threshold model. It is also interesting  to design approximation algorithms for other special kinds of complex graphs such as small-world graphs. Another interesting research direction is to study maximum active set problem for non-progressive models. 

%% file: appendix.tex
\appendix
\section{Proofs from Section \ref{sec:general}} \label{app:general_proofs}
In this section we give the missing proofs from Section \ref{sec:general}.
\\
\\
\noindent {\bf Proof of Lemma \ref{thm:bip_to_gen}.}
Consider a graph $G$ with $n$ vertices and vertex set $V(G)=\{v_1,v_2,\ldots v_n\}$ and threshold function $t$. Assume that there is a Perfect Target Set $f_0$ for $G$ such that $cost(f_0) < \alpha |V(G)| $. Let $H=(X,Y)$ be a bipartite graph such that $X=\{x_1,\ldots x_n\}$ and $Y=\{y_1,\ldots y_n\}$ and $t'$ be the threshold function of vertices of $H$ such that for every $1\leq i \leq n$, $t'(x_i)=t'(y_i)=t(v_i)$. Define $E(H)=\{x_{i}y_{j} | v_{i}v_{j} \in E(G) \}$. Let $g_0$ be a Target Set for $H$ such that $g_0(x_i)=g_0(y_i)=f_0(v_i)$ for every $1 \leq i \leq n$. We claim that $g_0$ is a PTS for $H$. By induction on $\tau$, we prove that $g_{\tau}(x_i)=g_{\tau}(y_i)=f_{\tau}(v_i)$ for every $1 \leq i \leq n$. By the definition, the assertion is true for $\tau=0$. Now let the assertion be true for time $\tau$. Consider a vertex $x_i\in X$. We have $\sum_{y\in N(x_i)} g_{\tau}(y)=\sum_{v\in N(v_i)} f_{\tau}(v)$ and also $t(x_i)=t(v_i)$, thus $x_i$ is influenced at time $\tau+1$ by $g_0$ iff $v_i$ is influenced at time $\tau+1$ by $f_0$. By similar justification we can show that $g_{\tau+1}(y_i)=f_{\tau+1}(v_i)$ too. So $g_0$ is a PTS for $H$ iff $f_0$ is a PTS for $G$, which is a contradiction since by assumption $NPPTS(H)\geq \alpha |V(H)|$ but $cost(g_0)<\alpha |V(H)|$.

Now we prove that $NPPTS(G)\leq \beta |V(G)|$. Consider the bipartite graph $H$ with the aforementioned definition. By assumption there is a Perfect Target Set $g_0'$ with weight at most $\beta |V(H)|$ for $H$. With no loss of generality assume that the number of vertices in $X$ for which $g_0'$ is one (initially infected vertices) is less than the number of initially infected vertices of $Y$. Let $f_0'$ be a PTS for $G$ such that $f_0'(v_i)=g_0'(x_i)$ for every $1 \leq i \leq n$. We have $cost(g_0')\leq \beta |V(G)|$ since $|V(H)|=2|V(G)|$. By induction on $\tau$ we show that $f_{2\tau}'(v_i)=g_{2\tau}'(x_i)$ and $f_{2\tau+1}'(v_i)=g_{2\tau+1}'(y_i)$ for every $1\leq i \leq n$ and every $\tau \geq 0$. The assertion is trivial for $\tau=0$. Now let the assertion be true for time $2\tau$. Consider a vertex $v_i\in V(G)$. We have $\sum_{v\in N(v_i)} f_{2\tau}'(v)=\sum_{x\in N(y_i)} g_{2\tau}'(x)$ and also $t(v_i)=t(y_i)$, thus $v_i$ is influenced at time $2\tau+1$ by $f_0'$ iff $y_i$ is influenced at time $2\tau+1$ by $g_0'$. By similar justification we can show that $f_{2\tau+2}'(v_i)=g_{2\tau+2}'(x_i)$ too. So $g_0'$ is a PTS for $H$ iff $f_0'$ is a PTS for $G$ and so $NPPTS(G) \leq \beta |V(G)|$.
\\
\\
In the following, $d_{S}(v)$ denotes the number of neighbors of $v$ in subset. 
\\
\\
\noindent {\bf Proof of Lemma \ref{lem:maxdegree}.}
Consider a set $S\subseteq V(G)$. With no loss of generality, suppose that $f_0(v)=0$ for every $v\in S \cap X$. We prove the lemma by contradiction. Assume that for every $u\in S$, $d_S(u) > d(u) - t(u) $. For every $y\in S\cap Y$, $f_1(y)=0$ since $y$ has at least $d(y) - t(y) + 1$ adjacent vertices in $S\cap X$ for which $f_0$ is zero. Similarly, for every $x\in S \cap X$, $f_2(x)=0$ since $x$ has at least $d(x) - t(x) + 1$ adjacent vertices in $S \cap Y$ for which $f_1$ is zero, and so on. Thus $f_0$ is not a Perfect Target Set, a contradiction.
\\
\\
\noindent {\bf Proof of Lemma \ref{lem:maxedge}.}
We prove the lemma by induction on $n$. For $n=1$ the assertion is trivial. Consider a graph $G$ with $n$ vertices. Let $S=V(G)$. By assumption, there is at least one vertex $v$, such that $d(v) \leq d(v) - t(v)$. Remove the vertex $v$ from $G$. By induction hypothesis $G - v$ has at most $\sum_{u\in V(G-v)} (d(u) - t(u))$ edges, so $G$ has at most $\sum_{u\in V(G)} (d(u) - t(u))$ edges.
\\
\\
\noindent {\bf Proof of Theorem \ref{thm:tightlower}.}
Consider a $(d+1)$-regular graph $G_{1}$ with $m_{1}$ vertices . In step $i$ ($1 \leq i \leq \infty)$, Add $m_{i+1} = \frac{d}{d+1}m_{i}$ vertices to the graph and connect each of them to $G_{i}$ by $d+1$ edges. Each vertex of $G_{i}$ must receive exactly $d$ newly edges. Name the subgraph formed by these vertices $G_{i+1}$. This process is shown in Figure \ref{fig:tightlower}. The final graph has $n = \sum_{i=1}^{\infty}{m_{i}} = m_{1}(d+1)$ vertices. It is easy to show that $V(G_{1})$ is a PTS, so $NPPTS(G) \leq |V(G_{1})| = m_{1} = \frac{2n}{2d+2} = \frac{2n}{\Delta + 1 }$.
\begin{figure}
\begin{center}
\includegraphics[width=6cm]{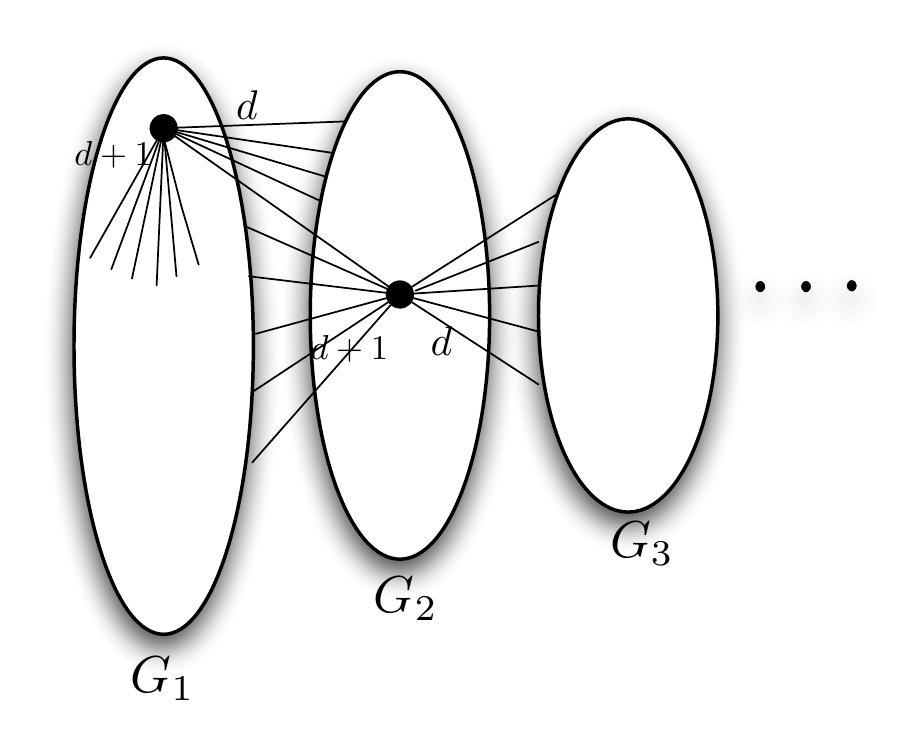}
\end{center}
\caption{A tight example for $NPPTS(G)$'s lower bound}
\label{fig:tightlower}
\end{figure} 
\\
\\
\begin{figure}
\begin{center}
\includegraphics[width=8cm]{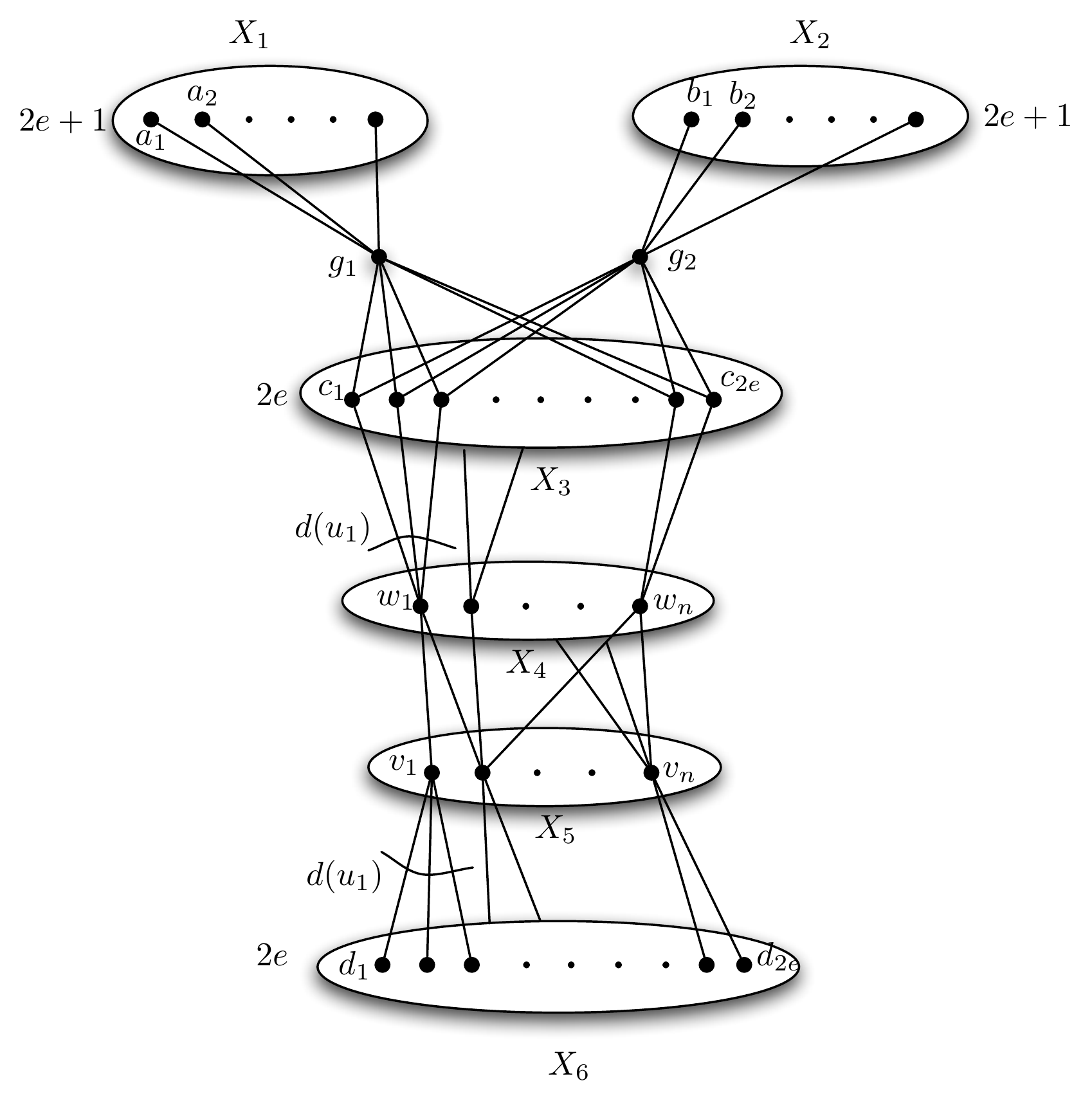}
\end{center}
\caption{The graph $H$}
\label{fig:hardness}
\end{figure} 
\noindent {\bf Proof of Theorem \ref{thm:nphardness}.}
In an instance of the minimum dominating set problem (MDS), given a graph $G(V, E)$, our goal is to find find a subset $S\subseteq V(G)$ of minimum cardinality such that for any node $v\not \in S$, we have $S\cap N(v)\not = \emptyset$. We give reduction from this NP-hard problem to our problem. 
Given an instance $G$ of MDS with $V(G) = \{u_{1}, u_{2},...,u_{n}\}$ and $|E(G)| = e$, we define an undirected graph $H$ as follows (See Figure \ref{fig:hardness}): First, let
\begin{align*}
X_{0} & = \{g_{1},g_{2}\} &  X_{1} & = \{a_{i}| 1 \leq i \leq 2e+1\} \\
X_{2} & = \{b_{i}| 1 \leq i \leq 2e+1\} &  X_{3} & = \{c_{i}| 1 \leq i \leq 2e\} \\
X_{4} & = \{w_{i}| 1 \leq i \leq n\} & X_{5} & = \{v_{i}| 1 \leq i \leq n\} \\
X_{6} & = \{d_{i}| 1 \leq i \leq 2e\}.\\
\end{align*}
Now let $H(V, E)$ be 
\begin{align*}
V(H) & = \cup_{i=0}^{6}{X_{i}} \\
\end{align*}
\begin{align*}
E(H) & = \{g_{1}a_{i}| 1 \leq i \leq 2e+1\} \\
	& \cup \hspace{0.1cm}\{g_{2}b_{i}| 1 \leq i \leq 2e+1\} \\
	& \cup \hspace{0.1cm} \{g_{1}c_{i}| 1 \leq i \leq 2e\} \\
	& \cup \hspace{0.1cm} \{g_{2}c_{i}| 1 \leq i \leq 2e\} \\
	& \cup  \hspace{0.1cm} \{w_{i}c_{j}|  1 \leq i \leq n, \sum_{k=1}^{i-1}{d(u_{k})} \leq j \leq \sum_{k=1}^{i}{d(u_{k})} \} \\
	& \cup  \hspace{0.1cm} \{v_{i}w_{j}| u_{i}u_{j} \in E(G) \vee i = j \} \\
	& \cup \hspace{0.1cm} \{v_{i}d_{j}|  1 \leq i \leq n, \sum_{k=1}^{i-1}{d(u_{k})} \leq j \leq \sum_{k=1}^{i}{d(u_{k})} \}. \\
\end{align*}
Suppose that $D$ is a minimum dominating set for $G$. Define $D^{H} = \{v_{i}|u_{i}\in D\}$.  We show that $NPPTS(G) = 2e+n+4+|D|$. It is easy to see that vertices in $X_{0} \cup X_{3} \cup X_{4} \cup D^{H}$ plus one vertex from each of $X_{1}$ and $X_{2}$ form a Perfect Target Set for the graph $H$. So, we have $NPPTS(H) \leq |X_{0}| + |X_{3}| + |X_{4}| + |D^{H}|+ 2 =  2e+n+4+|D|$. 

It remains to prove that $NPPTS(H) \geq 2e+n+4+|D|$. Suppose that $S \subseteq V(H)$ is a PTS for $H$ with minimum cardinality. Consider vertex $g_{1}$ in time $\tau$. If $f_{\tau}(g_{1}) = 0$, in time $\tau+1$ for every vertex $a_{i} \in X_{1}$ we will have $f_{\tau+1}(a_{i}) = 0$ and then $f_{\tau+2}(g_{1}) = 0$. So we have, $g_{1} \in S$. Similarly, we have $g_{2} \in S$. Moreover, at least $2e+1$ vertices from each of $g_{1}$ or $g_{2}$'s neighbors must be in $S$, so w.l.o.g suppose that $X_{3}$'s members plus at least one vertex from each of $X_{1}$ and $X_{2}$ are in $S$. By this setting, the vertices of $X_{0} \cup X_{1} \cup X_{2} \cup X_{3}$ become infected and keep this infection for every $\tau > 0$.

Consider a vertex $w_{k} \in X_{4}$. Let  $B(w_{k}) = \{d_{i} \in X_{6}| d_{i} $ is reachable from $w_{k}$ by a path of length 2$\}.$ Suppose that $w_{k} \notin S$. If there exists a $d_{i} \in B(w_{k}) \cap S$, we replace it by $w_k$ in $S$. This modification does not prevent $S$ from being a PTS and also does not increase $|S|$. So, we may assume that $B(w_{k})\cap S = \emptyset$ when $w_k\notin S$. Now, consider one of $w_{k}$'s neighbors in $X_{5}$ such as $v_{p}$. None of $v_{p}$'s neighbors in $X_{6}$ are infected initially. Thus $v_{p}$ has at most $d(u_{p})$ initially infected neighbors. this implies that $f_{1}(v_{p}) = 0$ and it is true for all other $w_{k}$'s neighbors in $X_{5}$. Similarly, $f_{2}(w_{k}) = 0$ and $f_{2}(d_{j}) = 0$ for all $d_{j} \in B(w_{k})$. 
Similar to this argument, one can show  that for every $\tau > 0$, $f_{2\tau}(w_{k}) = 0$ and $f_{2\tau}(d_{j}) = 0$ for all $d_{j} \in B(w_{k)}$. Therefore, for every $w_{k} \in X_{4}$, at least one of its neighbors in $X_{5}$ must be in $S$. This means that $S\cap X_{5}$ must have at least $|D|$ vertices and the proof is complete.


\section{Experimental Evaluation Data on Real Networks} \label{app:experiments_table}

\noindent {\bf Generating andom power-law networks.} 
We evaluate the performance of the greedy  algorithm on graphs with various amount of power-law coefficient. Following a previously developed way of generating power-law graphs from \cite{aiello2000random}, we set two parameters $\alpha$ and $\gamma$ defined as follows: $\alpha$ is the logarithm of the graph size and $\gamma$ is the log-log growth rate (power-law coefficient). 
The number of vertices with degree $x$, $y$  satisfies 
$$\log y = \alpha - \gamma \log x.$$
The random power-law graph model is defined as follows: given $n$ weighted vertices with weights $w_{1}, w_{2}, \cdots, w_{n}$, a pair $(i,j)$ of vertices appears as an edge with probability 
$w_{i}w_{j}p$ independently. These parameters $p$ and $w_{1}, w_{2}, \cdots, w_{n}$ must satisfy 
\begin{itemize}
\item $\sharp\{i|w_{i} = 1\} = \lfloor e^{\alpha} \rfloor - r$ and $\sharp\{i|w_{i} = k\} = \lfloor \frac{e^{\alpha}}{k^{\gamma}} \rfloor$ for $k = 2,3,.., \lfloor e^{\frac{\alpha}{\gamma}} \rfloor$. Here $\alpha$ is a value minimizing $|n - \sum_{k=1}^{\lfloor e^{\frac{\alpha}{\gamma}} \rfloor}{\lfloor \frac{e^{\alpha}}{k^{\gamma}} \rfloor}|$ and $r = n - \sum_{k=1}^{\lfloor e^{\frac{\alpha}{\gamma}} \rfloor}{\lfloor \frac{e^{\alpha}}{k^{\gamma}} \rfloor}$.
\item $p = \frac{1}{\sum_{i=1}^{n}{w_{i}}}$
\end{itemize}
One can easily see the expected degree of $i$'th vertex would be $w_{i}$ and also vertices' weights follow power-law. 

{\noindent \bf Experimental results for four real-world networks.}
Table \ref{t:realnetworkdata} includes the exact amount of greedy NPPTS's output compared to the output of other heuristics. 
\begin{table*}[htpb]
\caption{\textbf{Results on the real networks}}
\label{t:realnetworkdata}
\centering
\begin{tabular}{|p{5cm}|c|c|c|c|c|c|} 
\hline
Network  & No. of & $\gamma$ & \multicolumn{4}{c}{ No. of nodes selected by algorithm} \vline \\
	& nodes & & Greedy & High Degree & Central & Random  \\ 

\hline\hline
Who-trusts-whom network of Epinions.com & 75888 & 1.50 & 27131 & 75878 & 75879 & 75888 \\
\hline
Slashdot social network & 77360 & 1.68 & 49978 & 77327 & 77360 & 77360 \\
\hline
Collaboration network of Arxiv Astro Physics & 18772 & 1.84 & 8287 & 18771 & 18772 & 18763 \\
\hline
Arxiv High Energy Physics paper citation network & 34546 & 2.05 & 14647 & 34539 & 34546 & 34505
 \\
\hline

Amazon product co-purchasing network  & 262111 & 2.54 & 155085 & 262111 & 262005 & 262026
 \\
\hline
\end{tabular}\label{table:experimentals}
\end{table*}

\section{Missed Things from Section \ref{sec:convergence}} \label{app:convergence}
{\noindent \bf Proof of theorem \ref{t:convergence}.}
In \cite{goles1980periodic}, it is shown that, for a function $\Delta$ from $\{0, 1\}^n$ to $\{0, 1\}^n$ whose components from a symmetric set of threshold functions, the repeated application of $\Delta$, leads either to a fixed point or to a cycle of length two. Since the set of functions $f_{\tau}$ (defined in Section \ref{s:introduction}) are symmetric threshold functions, the lemma follows immediately from this fact.
\\
{\noindent \bf Average convergence time of the process on social networks.} Applying  a sampling technique and simple concentration inequalities, one can easily show that  the average convergence time of the non-progressive process on  graph $G$ can be approximated with an additive error of $\epsilon$  in time $O(\frac{e^{2}.n\log(n)}{\epsilon^{2}})$ where $e = |E(G)|$ and $n = |V(G)|$.
\begin{theorem}
Computing the average convergence time of the non-progressive process on  graph $G$, with an error of $\epsilon$ is possible in time $O(\frac{e^{2}.n\log(n)}{\epsilon^{2}})$ where $e = |E(G)|$ and $n = |V(G)|$. 
\end{theorem}
\begin{proof}
Define random variable $X_{S} = ct_{G}(S)$. We uniformly select some of the $V(G)$'s subsets $S_{1}, S_{2}, ..., S_{m}$ and take the average of $X_{S_{i}}$s. In \cite{hoeffding1963probability}, Hoeffding shows that with large value of $m$ and if $X_{S_{i}}$s are bounded between $a_{i}$ and $b_{i}$, $\overline{X_{S}}$ would be a good estimation (with an error less than $\epsilon$) for $E[X_{S}]$ that is our desired target:
$$Pr(|\overline{X_{S}} - \mathrm{E}[\overline{X_{S}}]| \geq \epsilon) \leq 2\exp \left( - \frac{2\epsilon^2m^2}{\sum_{i=1}^m (b_i - a_i)^2} \right)$$
From Theorem \ref{thm:convergence} we know putting $a_{i} = 0$ and $b_{i} = 8e$ for all $1 \leq i \leq m$, meets the preconditions of the above inequality. To have $Pr(|\overline{X_{S}} - \mathrm{E}[\overline{X_{S}}]|) \leq \frac{2}{n}$, we can set 
$$m^{2} \geq \frac{ln(n)\sum_{i=1}^m (b_i - a_i)^2}{\epsilon^{2}} = \frac{64.m.e^{2}.ln(n)}{\epsilon^{2}} \Rightarrow m \geq \frac{64e^{2}.ln(n)}{\epsilon^{2}}$$
Since computing each $X_{S_{i}}$ needs $O(n)$ the total time will be at most $O(mn) = O(\frac{e^{2}.n\log(n)}{\epsilon^{2}})$.
\end{proof}
\begin{corollary}
Computing the average convergence time of the process on a power-law graph $G$, with an error of $\epsilon$ is possible in time $O(\frac{n^{3}\log(n)}{\epsilon^{2}})$ where $n = |V(G)|$. 
\end{corollary}

 As a result, we can perform experimental evaluation of convergence time in several families of graphs. In particular, through experimental evaluations, we show the average time of convergence for random power law graphs with $\epsilon = 0.1$. Figure \ref{f:convergencechart} shows average convergence time calculated by sampling for $500$ random power law graphs with average of $100$ vertices.

\begin{figure}
\centering
\includegraphics[width=5cm, trim=2cm 2cm 7cm 17cm]{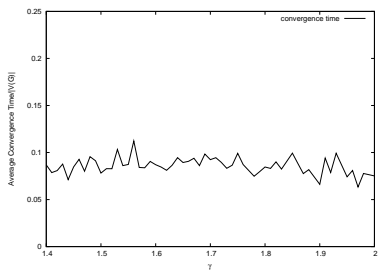}
\caption{Average convergence time on random power-law graphs }
\label{f:convergencechart}
\end{figure}